\documentclass[letterpaper, 10 pt, conference]{ieeeconf}
\usepackage{amsfonts}
\usepackage{amsmath,amssymb}
\usepackage{cite}
\usepackage{graphicx}
\usepackage{hyperref}
\usepackage{wrapfig}
\usepackage{times}

\newtheorem{theorem}{Theorem}

\newtheorem{corollary}[theorem]{Corollary}

\newtheorem{definition}[theorem]{Definition}

\newtheorem{lemma}[theorem]{Lemma}

\newtheorem{problem}[theorem]{Problem}

\newtheorem{remark}[theorem]{Remark}

\IEEEoverridecommandlockouts                              
\newcommand{\im}{\mathbf{j}}

\begin{document}

\title{{\LARGE \textbf{Reconstruction of Directed Networks from Consensus
Dynamics}}}
\author{Shahin Shahrampour and Victor M. Preciado \thanks{$^{\dagger }$The authors are with the
Department of Electrical and Systems Engineering, University of
Pennsylvania, Philadelphia, PA 19104-6228 USA. \texttt{\small %
\{shahin,preciado\}@seas.upenn.edu}}}
\maketitle

\begin{abstract}
This paper addresses the problem of identifying the topology of an unknown,
weighted, directed network running a consensus dynamics. We propose a methodology to
reconstruct the network topology from the dynamic response when the system is stimulated by a wide-sense stationary noise of unknown power spectral density. The method is based on a node-knockout, or grounding, procedure wherein the grounded node broadcasts zero without being eliminated from the network. In this direction, we measure the empirical cross-power spectral densities of the
outputs between every pair of nodes for both grounded and ungrounded consensus to reconstruct the unknown topology of the network. We also establish that in the special cases of undirected and purely unidirectional networks, the reconstruction does not need grounding. Finally, we extend our results to the case of a directed network assuming a general dynamics, and prove that the developed method can detect edges and their direction. 
\end{abstract}

\thispagestyle{empty} \pagestyle{empty}



\section{Introduction}

In recent years, complex dynamical networks have attracted considerable
attention \cite{strogatz2001exploring}. The power grid, the Internet, the
World Wide Web, as well as many other biological, social and economic
networks \cite{jackson2008social}, are examples of networked dynamic systems
that motivate this interest. The availability of datasets describing the
structure of many real-world networks has allowed to detect the presence of
common patterns in a large variety of networks \cite%
{newman2003structure,boccaletti2006complex}. In this paper, we address the
problem of reconstructing the structure of an unknown network of dynamical
nodes from observations of its input-output behavior.

The problem of network reconstruction is crucial in a wide variety of
disciplines such as biology\cite%
{bonneau2006inferelator,geier2007reconstructing,bansal2007infer,julius2009genetic}%
, physics\cite%
{boccaletti2007detecting,timme2007revealing,napoletani2008reconstructing}
and finance\cite{mantegna2000introduction}. A wide collection of approaches
have been proposed to solve the network reconstruction problem. For example,
we find in the literature several papers that approach this problem using an
optimization framework, such as \cite{julius2009genetic}, \cite%
{napoletani2008reconstructing}, \cite{candes2008enhancing}. In these papers,
the reconstruction problem is stated as the optimization problem of finding
the network that maximizes a function that measures the sparsity of the
network (e.g. $\ell_{1}$-norm) while conforming to known a priori structural
information. Although the assumption of sparsity is well justified in some
applications (e.g. biological networks), this assumptions might lead to
unsuccessful topology inference in other cases, as illustrated in \cite%
{gonccalves2008necessary} and \cite{yuan2011robust}. When the unknown
network is known to be a tree, several techniques for network reconstruction
were proposed in \cite{mantegna2000introduction}, \cite%
{materassi2009unveiling} and \cite{materassi2010topological}. More recently,
Materassi and Salapaka proposed in \cite{materassi2012problem} a methodology
for reconstruction of directed networks using Wiener filters. Although
effective for many networks, this methodology is not exact when two
nonadjacent nodes are connected to a common node with directed edges
pointing towards the common node. In \cite{nabi2011sieve}, Nabi-Abdolyousefi
and Mesbahi proposed a technique to extract structural information, such as
node degrees, of an undirected network running a consensus dynamics.
Subsequently, they find a collection of undirected graphs that are
consistent with this structural information. Furthermore, the work in \cite%
{6203379} provides a method that performs a complete identification of
undirected networks via a procedure called node knock-out.

In this paper, we propose an approach to reconstruct the structure and weights of a directed network from the output of an agreement dynamics run on the network. We also revisit the problem of network reconstruction for more general dynamical systems. We
develop methodologies to unveil the network structure from the dynamic
response of the network when the system is driven by stochastic inputs. More specifically, we assume that the system is stimulated by a collection of wide-sense stationary noises with
unknown power spectral densities. Considering several cases, we
propose methodologies to recover the network topology from the empirical
cross-power spectral densities of the outputs between every pair of nodes. We first consider the case of undirected networks running a consensus dynamics, and propose an algorithm to reconstruct its unknown topology. Using the node knockout procedure proposed in \cite{6203379}, we extend our results to the directed networks. In this scenario, the node
knockout is equivalent to state grounding, where the node
broadcasts a zero state to its neighbors without being eliminated 
from the network. We also prove that for purely unidirectional networks (networks with no reciprocity \cite{NewmanBOOK}), there is no need to run the grounded consensus to perform the network identification. Finally, we consider the reconstruction of directed networks assuming a general dynamics. We establish that, without the knowledge of the power spectral densities of the input vector, a Boolean reconstruction is still possible, i.e., we can detect edges and their direction, but not their weights.

The rest of the paper is organized as follows. In section II, we provide
some nomenclature needed in our exposition and describe the network
reconstruction problem under consideration. Assuming the system is driven by wide-sense stationary noise, in section III.A, we provide a technique to recover the structure of an undirected network running
consensus dynamics.
Section III.B covers the extension of the problem to a directed network. Section III.C addresses the
reconstruction of directed networks following general dynamics. Section IV concludes.

\section{Preliminaries \& Problem Description}

In this section, we introduce a series of definitions used throughout the
paper. Let $\mathcal{G}=\left( \mathcal{V},\mathcal{E}\right) $ be an
unweighted, undirected graph, where $\mathcal{V}=\left\{ v_{1},\dots
,v_{n}\right\} $ denotes a set of $n$ nodes and $\mathcal{E}\subseteq 
\mathcal{V}\times \mathcal{V}$ denotes a set of $m$ undirected edges. If $%
\left\{ v_{i},v_{j}\right\} \in \mathcal{E}$, we call nodes $v_{i}$ and $%
v_{j}$ \emph{adjacent} (or first-neighbors), which we denote by $v_{i}\sim
v_{j}$.\ A \emph{weighted}, undirected graph is defined as the triad $%
\mathcal{W=}\left( \mathcal{V},\mathcal{E},\mathcal{F}\right) $, where $%
\mathcal{V}$ and $\mathcal{E}$ are the sets of nodes and edges in $\mathcal{W%
}$, and the function $\mathcal{F}:\mathcal{E\rightarrow }\mathbb{R}$
associates real weights to the edges. Similarly, a weighted, \emph{directed}
graph is defined as the triad $\mathcal{D=}\left( \mathcal{V},\mathcal{E}%
_{d},\mathcal{F}_{d}\right) $, where $\mathcal{V}$ is the set of nodes and $%
\mathcal{E}_{d}$ is the set of directed edges in $\mathcal{D}$, where a
directed edge from node $v_{i}$ to node $v_{j}$ is defined as an ordered
pair $\left( v_{i},v_{j}\right) $. Furthermore, $\mathcal{F}_{d}$ is a
weight function $\mathcal{F}_{d}:\mathcal{E}_{d}\rightarrow \mathbb{R}$.

In an unweighted, undirected graph $\mathcal{G}$, the \emph{degree} of a
vertex $v_{i}$, denoted by $\deg \left( v_{i}\right) $, is the number of
nodes adjacent to it, i.e., $\deg \left( v_{i}\right) =\left\vert \left\{
v_{j}\in \mathcal{V}:\left\{ v_{i},v_{j}\right\} \in \mathcal{E}\right\}
\right\vert $. This definition can be generalized to both weighted and
directed graphs. For weighted graphs, the weighted degree of node $v_{i}$ is
equal to $\deg \left( v_{i}\right) =\sum_{j:\left\{ v_{i},v_{j}\right\} \in 
\mathcal{E}}\mathcal{F}\left( \left\{ v_{i},v_{j}\right\} \right) $, i.e.,
the sum of the weights associated to edges connected to $v_{i}$. For
weighted, directed networks, we define the weighted \emph{in-degree} of node 
$v_{i}$ as $\deg _{in}\left( v_{i}\right) =\sum_{j:\left( v_{j},v_{i}\right)
\in \mathcal{E}_{d}}\mathcal{F}_{d}\left( \left( v_{j},v_{i}\right) \right) $%
.

The \emph{adjacency matrix} of an unweighted, undirected graph $\mathcal{G}$%
, denoted by $A_{\mathcal{G}}=[a_{ij}]$, is a $n\times n$ Boolean symmetric
matrix defined entry-wise as $a_{ij}=1$ if nodes $v_{i}$ and $v_{j}$ are
adjacent, and $a_{ij}=0$ otherwise. We define the \emph{Laplacian matrix }$%
L_{\mathcal{G}}$ of a graph $\mathcal{G}$ as $L_{\mathcal{G}}=D_{\mathcal{G}%
}-A_{\mathcal{G}}$ where $D_{\mathcal{G}}$ is the diagonal matrix of
degrees, $D_{\mathcal{G}}=diag\left( \left( \deg \left( v_{i}\right) \right)
_{i=1}^{n}\right) $. For simple graphs, $L_{\mathcal{G}}$ is a symmetric
positive semidefinite matrix, which we denote by $L_{\mathcal{G}}\succeq 0$ 
\cite{biggs1993algebraic}. Thus, $L_{\mathcal{G}}$ has a full set of $n$
real and orthogonal eigenvectors with real nonnegative eigenvalues $%
0=\lambda _{1}\leq \lambda _{2}\leq ...\leq \lambda _{n}$.

Similarly, the weighted adjacency of a weighted graph $\mathcal{W}$ is
defined as $A_{\mathcal{W}}=\left[ w_{ij}\right] $, where $w_{ij}=\mathcal{F}%
\left( \left\{ v_{i},v_{j}\right\} \right) $ for $\left\{
v_{i},v_{j}\right\} \in \mathcal{E}$, and $w_{ij}=0$ if $\left\{
v_{i},v_{j}\right\} \not\in \mathcal{E}$. We define the \emph{degree matrix}
of a weighted graph $\mathcal{W}$ as the diagonal matrix $D_{\mathcal{W}%
}=diag\left( \left( \deg \left( v_{i}\right) \right) _{i=1}^{n}\right) $.
The Laplacian matrix of a weighted, undirected graph $\mathcal{W}$, is
defined as $L_{\mathcal{W}}=D_{\mathcal{W}}-A_{\mathcal{W}}$. Furthermore,
the adjacency matrix of a weighted, directed graph $\mathcal{D}$ is defined
as $A_{\mathcal{D}}=\left[ d_{ij}\right] $, where $d_{ij}=\mathcal{F}%
_{d}\left( \left( v_{j},v_{i}\right) \right) $ for $\left(
v_{j},v_{i}\right) \in \mathcal{E}_{d}$, and $d_{ij}=0$ if $\left(
v_{j},v_{i}\right) \not\in \mathcal{E}_{d}$. We define the \emph{in-degree
matrix} of a directed graph $\mathcal{D}$ as the diagonal matrix $D_{%
\mathcal{D}}=diag\left( \left( \deg _{in}\left( v_{i}\right) \right)
_{i=1}^{n}\right) $. The Laplacian matrix of $\mathcal{D}$ is then defined
as $L_{\mathcal{D}}=D_{\mathcal{D}}-A_{\mathcal{D}}$. Note that the
Laplacian matrix satisfies $L_{\mathcal{G}}\mathbf{1}=L_{\mathcal{W}}\mathbf{%
1}=L_{\mathcal{D}}\mathbf{1}=\mathbf{0}$, i.e., the vector of all ones is an
eigenvector of the Laplacian matrix with corresponding eigenvalue $0$.

In this paper, we focus on the continuous-time non-autonomous model of
consensus dynamics in directed networks, described as
\begin{equation}
\dot{x}(t)=-L_{\mathcal{D}}x(t)+\mathbf{w}\left( t\right) \ ,\ y(t)=x(t),  \label{dynamics}
\end{equation}
where $L_{\mathcal{D}}$ is the Laplacian matrix of a weighted, directed
network $\mathcal{D}$; and the vectors $x(t),\mathbf{w}\left( t\right),y(t)\in \mathbb{R}^{n}$
are the state, input, and output vectors, respectively. We assume that we do
not have access to the network structure, i.e., $L_{\mathcal{D}}$ is
unknown. In this context, we consider the problem of identifying the
topology of the directed network $\mathcal{D}$ when the (cross-)power spectral densities of 
the output vector $y(t)$ are measured empirically while
the stochastic input vector $\mathbf{w}\left( t\right)$, injected at the nodes of
the network, has unknown power spectral characteristics.  We also investigate the reconstruction problem in the case that \eqref{dynamics} follows a general dynamics $G_{\mathcal{D}}$ in lieu of $-L_{\mathcal{D}}$. In any case, we assume the input $\mathbf{w}\left( t\right) =\left[ w_{i}\left( t\right) \right] $ is a vector of uncorrelated wide-sense stationary processes.

\begin{definition}[Wide-Sense Stationary]
A continuous-time scalar random process $w(t)$ is wide-sense stationary (WSS), if
it satisfies the following properties:

\begin{description}
\item[P1.] $\mu _{w}(t)\triangleq \mathbb{E}(w(t))=\mu _{w}(t+\tau )$ for
any $\tau \in \mathbb{R}$.\\

\item[P2.] $R_{w}(t_{1},t_{2})\triangleq \mathbb{E}%
(w(t_{1})w(t_{2}))=R_{w}(t_{1}+\tau ,t_{2}+\tau )=R_{w}(t_{1}-t_{2},0)$ for
any $\tau \in \mathbb{R}$.\\
\end{description}
\end{definition}

The reconstruction methods proposed in this paper take the output vector $%
y\left( t\right) =\left[ y_{i}\left( t\right) \right] $ and deliver the
network structure as a function of the (cross-)power spectral densities.

\begin{definition}[\textit{(Cross-)Power Spectral Density}]
The cross-power spectral density of two WSS signals, $y_{i}\left( t\right) $
and $y_{j}\left( t\right) $, is the Fourier transform of their
cross-correlation function, i.e., 
\begin{equation*}
S_{y_{i}y_{j}}(\omega )\overset{\Delta }{=}\mathcal{F}\{R_{y_{i}y_{j}}(\tau )%
\overset{\Delta }{=}\mathbb{E}(y_{i}(t)y_{j}(t-\tau ))\},
\end{equation*}%
where $\mathcal{F}\left\{ \cdot \right\} $ is the Fourier transform
operator. The power spectral density of $y_{i}\left( t\right) $ is defined
as 
\begin{equation*}
S_{y_{i}}(\omega )\overset{\Delta }{=}\mathcal{F}\{R_{y_{i}}(\tau )\overset{%
\Delta }{=}\mathbb{E}(y_{i}(t)y_{i}(t-\tau ))\}.\\
\end{equation*}
\end{definition}

\bigskip

In order to detect the links through observation of outputs, we need to ascertain that system \eqref{dynamics} is driven by a collection of noises with nonzero power spectral densities. To capture this idea, we   
commence with the following definition:

\begin{definition}[\textit{Excitation Frequency Interval}]
The excitation frequency interval of a vector $\mathbf{w}\left( t\right) $
of wide-sense stationary processes is defined as an interval $(-\Omega
,\Omega )$, with $\Omega >0$, such that the spectral densities of the input
components $w_{i}\left( t\right) $ satisfy $S_{w_{i}}(\omega )>0$ for all $%
\omega \in (-\Omega ,\Omega )$, and all $i\in \{1,2,...,n\}$.\\
\end{definition}

We end this section by stating our assumptions. Throughout the paper we impose the following conditions on the input vector:

\begin{description}
\item[A1.] The collection of signals $\left\{ w_{i}(t),i=1,...,n\right\} $
are uncorrelated zero-mean WSS processes such that, for any $t,\tau \in 
\mathbb{R}$,%
\begin{equation*}
\mathbb{E}(w_{i}(t)w_{j}(\tau ))=0,\text{ for }i\neq j,
\end{equation*}
and%
\begin{equation*}
R_{w_{i}}(\tau )=\mathbb{E}(w_{i}(t)w_{i}(t+\tau ))=R_{w}(\tau ).
\end{equation*}

\item[A2.] There exists a nonempty excitation frequency interval $\left(
-\Omega ,\Omega \right) $.
\end{description}

\section{Main Results}

We now consider several cases and present methodologies to reconstruct the structure of an
unknown network from observations of its temporal response \eqref{dynamics}.
We provide reconstruction techniques for weighted, undirected networks (Sect.
III.A), and weighted, directed networks (Sect. III.B) following consensus dynamics. We briefly, revisit the reconstruction of directed networks following general dynamics (Sect. III.C).

For the consensus dynamics (as well as general dynamics) of weighted, directed networks, the following
lemma provides an explicit relationship between the cross-power spectral
densities of two outputs, $y_{i}\left( t\right) $ and $y_{j}\left( t\right) $%
, and the power spectral density of the input $w_{k}\left( t\right) $
(which we assume to be identical, i.e., $S_{w_{k}}\left( \omega \right)
=S_{w}\left( \omega \right) $ for all $k$):

\begin{lemma}
\label{Power Spectral Densities}Given assumptions (A1)-(A2), the following
identity holds %
\begin{equation}
\frac{S_{y_{i}y_{j}}(\omega )}{S_{w}(\omega )}=\mathbf{e}_{i}^{T}(\omega
^{2}I-\im \omega (L_{\mathcal{D}}- L_{\mathcal{D}}^{T})+L_{\mathcal{D}%
}^{T}L_{\mathcal{D}})^{-1}\mathbf{e}_{j},  \label{spectrum}
\end{equation}%
for any $\omega \in \left(
-\Omega ,\Omega \right) $, where 
\begin{equation*}
S_{y_{i}y_{j}}(\omega )\overset{\Delta }{=}\mathcal{F}\{R_{y_{i}y_{j}}(\tau
)\}\ \ \text{and}\ \ S_{w}(\omega )\overset{\Delta }{=}\mathcal{F}%
\{R_{w}(\tau )\},
\end{equation*}%
for any $1\leq i,j\leq n$.
\end{lemma}

\begin{proof}
The transfer function corresponding to the state-space equations
\eqref{dynamics} is $H\left(  \omega\right)  =(\im \omega I+L_{\mathcal{D}}%
)^{-1}$. The transfer function from the $k$-th input $w_{k}\left(  t\right)  $
to the $i$-th output $y_{i}\left(  t\right)  $ is defined as%
\[
H_{ki}(\omega)\overset{\Delta}{=}\mathbf{e}_{i}^{T}(\im \omega I+L_{\mathcal{D}}%
)^{-1}\mathbf{e}_{k}.
\]
Therefore, the power spectral density of the $i$-th output $y_{i}\left(
t\right)  $ when the input is $\mathbf{w}\left(  t\right)  =w_{k}\left(  t\right)
\mathbf{e}_{k}$ (i.e., a WSS noise on the $k$-th node) is equal to
\begin{equation}
S_{y_{i}}(\omega)=H_{ki}(\omega)H_{ki}^{\ast}(\omega)S_{w_{k}}(\omega
).\label{transfer1}%
\end{equation}
On the other hand, the transfer functions from input $w_{k}\left(  t\right)  $
to the outputs $y_{i}\left(  t\right)  $ and $y_{j}\left(  t\right)  $ are,
respectively, $Y_{i}\left(  \omega\right)  /W_{k}\left(  \omega\right)
=H_{ki}(\omega)$ and $Y_{j}\left(  \omega\right)  /W_{k}\left(  \omega\right)
=H_{kj}(\omega)$. Hence, $Y_{j}\left(  \omega\right)  /Y_{i}\left(
\omega\right)  =H_{ki}^{-1}(\omega)H_{kj}(\omega)$ which implies
\begin{equation}
S_{y_{i}y_{j}}(\omega)=\bigg(H_{kj}(\omega)H_{ki}^{-1}(\omega)\bigg)^\ast S_{y_{i}}%
(\omega).\label{transfer2}%
\end{equation}
Since $S_{w_{k}}(\omega)=S_{w}(\omega)$ for all $k$, we can combine
\eqref{transfer1} and \eqref{transfer2} to obtain%
\begin{equation}
S_{y_{i}y_{j}}(\omega)=H_{ki}(\omega)H_{kj}^{\ast
}(\omega)S_{w}(\omega).\label{superposition}%
\end{equation}
When the input is $\mathbf{w}\left(  t\right)=\sum_{k=1}^{n}w_{k}\left(  t\right)
\mathbf{e}_{k}$, with $\mathbb{E}(w_{i}(t)w_{j}(\tau))=0$ for $i\neq j$, we
can superpose \eqref{superposition} over $1\leq k\leq n$, to obtain the following for any $\omega \in \left(
-\Omega ,\Omega \right) $:

{\small
\begin{align*}
\frac{S_{y_{i}y_{j}}(\omega)}{S_{w}(\omega)} &  =\sum_{k=1}^{n}H_{kj}^{\ast
}(\omega)H_{ki}(\omega)\\
&  =\sum_{k=1}^{n}\mathbf{e}_{j}^{T}(-\im \omega I+L_{\mathcal{D}})^{-1}%
\mathbf{e}_{k}\mathbf{e}_{i}^{T}(\im \omega I+L_{\mathcal{D}})^{-1}\mathbf{e}%
_{k}\\
&  =\sum_{k=1}^{n}\mathbf{e}_{j}^{T}(-\im \omega I+L_{\mathcal{D}})^{-1}%
\mathbf{e}_{k}\mathbf{e}_{k}^{T}(\im \omega I+L_{\mathcal{D}}^{T})^{-1}%
\mathbf{e}_{i}\\
&  =\mathbf{e}_{j}^{T}(-\im \omega I+L_{\mathcal{D}})^{-1}(\sum_{k=1}%
^{n}\mathbf{e}_{k}\mathbf{e}_{k}^{T})(\im \omega I+L_{\mathcal{D}}^{T}%
)^{-1}\mathbf{e}_{i}.
\end{align*}
} As $\sum_{k=1}^{n}\mathbf{e}_{k}\mathbf{e}_{k}^{T}=I$, we can simplify the
last equation to
\begin{align*}
\frac{S_{y_{i}y_{j}}(\omega)}{S_{w}(\omega)} &  =\mathbf{e}_{j}^{T}(-\im \omega
I+L_{\mathcal{D}})^{-1}(\im \omega I+L_{\mathcal{D}}^{T})^{-1}\mathbf{e}_{i}\\
&  =\mathbf{e}_{i}^{T}(\omega^{2}I-\im \omega L_{\mathcal{D}}+\im \omega
L_{\mathcal{D}}^{T}+L_{\mathcal{D}}^{T}L_{\mathcal{D}})^{-1}\mathbf{e}_{j},
\end{align*}
which is the desired relation.
\end{proof}

\begin{corollary}\label{generalized lemma}
Given assumptions (A1)-(A2), and substituting $-L_{\mathcal{D}}$ in \eqref{dynamics} with any negative semi-definite $G_{\mathcal{D}}$, the input-output power spectra relationship is as following
\begin{equation*}
\frac{S_{y_{i}y_{j}}(\omega )}{S_{w}(\omega )}=\mathbf{e}_{i}^{T}(\omega
^{2}I-\im \omega (G_{\mathcal{D}}^T- G_{\mathcal{D}})+G_{\mathcal{D}%
}^{T}G_{\mathcal{D}})^{-1}\mathbf{e}_{j},  
\end{equation*}%
for any $\omega \in \left(
-\Omega ,\Omega \right) $.
\end{corollary}
\begin{proof}
In the proof of Lemma \ref{Power Spectral Densities}, we did not use any properties of Laplacian, so we only need to replace $-L_{\mathcal{D}}$ by any negative semi-definite $G_{\mathcal{D}}$ in \eqref{spectrum} to derive the result. 
\end{proof}

In our analysis in the next subsections, we frequently invoke the result of Lemma \ref{Power Spectral Densities} and Corollary \ref{generalized lemma}.

\subsection{Undirected Laplacian Identification from Stochastic Inputs}

In this subsection, we propose an approach to reconstruct the topology of an
unknown weighted, undirected network $\mathcal{W}$ from the output of %
\eqref{dynamics}.  The formal statement of the reconstruction problem for
weighted, undirected network can be stated as follows:

\begin{problem}
\label{Undirected Random Problem}Consider the dynamical network in %
\eqref{dynamics}, with $L_{\mathcal{D}}\equiv L_{\mathcal{W}}$ where $%
\mathcal{W}$ is an unknown weighted, undirected graph. Find the structure of 
$\mathcal{W}$, from the
empirical (cross-)power spectral densities\footnote{%
There are several methods to empirically measure the (cross-)power spectral
densities (see e.g., Welch's method\cite{welch1967use}).} of the outputs,
i.e., $S_{y_{i}}(\omega )$ and $S_{y_{i}y_{j}}(\omega )$ for all $1\leq
i,j\leq n$.
\end{problem}

In the case of weighted, undirected networks, we can use the result of Lemma \ref{Power Spectral Densities} in the form of the
following corollary:

\begin{corollary}
\label{Cross-Power Corollary}Consider the network dynamics \eqref{dynamics},
when $\mathcal{D}$ is a weighted, undirected network $\mathcal{W}$. Given
assumptions (A1)-(A2), the (cross-)power spectral densities satisfy%
\begin{equation}
\frac{S_{y_{i}y_{j}}(\omega )}{S_{w}(\omega )}=\mathbf{e}_{i}^{T}(\omega
^{2}I+L_{\mathcal{W}}^{2})^{-1}\mathbf{e}_{j},
\end{equation}%
for any $1\leq i,j\leq n$.
\end{corollary}

\begin{proof}
Since $L_\mathcal{W}$ is symmetric, replacing $L_\mathcal{D}$ in \eqref{spectrum} by $L_\mathcal{W}$, the proof follows immediately. 
\end{proof}

We now proceed with a theorem that can be used to reconstruct a
weighted, undirected Laplacian $L_{\mathcal{W}}$, from the empirical
(cross-)power spectral densities $S_{y_{i}}(\omega )$ and $%
S_{y_{i}y_{j}}(\omega )$ for all $1\leq i,j\leq n$.

\begin{theorem}
\label{Undirected Laplacian}Consider the network dynamics \eqref{dynamics},
when $\mathcal{D}$ is a weighted, undirected network $\mathcal{W}$. Let us
define the matrix of cross-correlations as $\mathbf{S}\left( \omega \right) =%
\left[ S_{y_{i}y_{j}}(\omega )\right] $. Then, given assumptions (A1)-(A2),
we can recover the weighted, undirected Laplacian $L_{\mathcal{W}}$ as%
\begin{equation}
L_{\mathcal{W}}=\omega \left( \frac{\mathbf{S}^{-1}\left( \omega \right) }{%
\left[ \mathbf{S}^{-1}(\omega )\mathbf{1}\right] _{i}}-I\right) ^{1/2}, \label{undlap}
\end{equation}%
for any $i$, and for any $\omega $ in the excitation frequency interval $%
\left( -\Omega ,\Omega \right) $.
\end{theorem}

\begin{proof}
According to Corollary \ref{Cross-Power Corollary}, we have
\[
\mathbf{S}(\omega)=S_{w}(\omega)(\omega^{2}I+L_{\mathcal{W}}^{2})^{-1},
\]
which yields
\begin{equation}
\mathbf{S}^{-1}(\omega)=\frac{1}{S_{w}(\omega)}(\omega^{2}I+L_{\mathcal{W}%
}^{2}).\label{lap}%
\end{equation}
Hence,%
\begin{equation}
L_{\mathcal{W}}^{2}=S_{w}(\omega)\mathbf{S}^{-1}(\omega)-\omega^{2}%
I.\label{Lw2}%
\end{equation}
We can derive an expression for $S_{w}(\omega)$ in terms of $\mathbf{S}%
^{-1}(\omega)$ and $\omega$, as follows. Post-multiplying \eqref{lap} by the
vector $\mathbf{1}$, we get
\begin{align}
\mathbf{S}^{-1}(\omega)\mathbf{1}=\frac{1}{S_{w}(\omega)}(\omega
^{2}I+L_{\mathcal{W}}^{2})\mathbf{1}=\frac{\omega^{2}}{S_{w}(\omega
)}\mathbf{1},\label{average}
\end{align}
since $L_{\mathcal{W}}^{2}\mathbf{1}=0$. Therefore,%
\begin{equation}
S_{w}(\omega)=\frac{\omega^{2}}{\left[  \mathbf{S}^{-1}(\omega)\mathbf{1}%
\right]  _{i}},\label{Sw}%
\end{equation}
for any $i$. Substituting (\ref{Sw}) in (\ref{Lw2}), we obtain the statement
of our Theorem. (Notice that $L_{\mathcal{W}}$ is the Laplacian of an
undirected network, hence, it is positive semidefinite).
\end{proof}

In practice, since $\left[ \mathbf{S}^{-1}(\omega )\mathbf{1}\right] _{i}$ is a noisy measurement, according to \eqref{average} one can average out noise by replacing $\left[ \mathbf{S}^{-1}(\omega )\mathbf{1}\right] _{i}$ with $\mathbf{1}^T \mathbf{S}^{-1}(\omega )\mathbf{1}/n$, to achieve improved empirical results. Hence, \eqref{undlap} can be rewritten as
\begin{equation*}
L_{\mathcal{W}}=\omega \left( \frac{\mathbf{S}^{-1}\left( \omega \right) }{%
\mathbf{1}^T \mathbf{S}^{-1}(\omega )\mathbf{1}}n-I\right) ^{1/2}. 
\end{equation*}%

\begin{remark}
To derive node degrees computationally efficient, one can circumvent the implicit eigenvalue decomposition involved in the result of Theorem\ref{Undirected Laplacian} for computing $L_{\mathcal{W}}$ from $L^2_{\mathcal{W}}$, since 
\begin{align*}
[L^2_{\mathcal{W}}]_{ii}&=[(D_{\mathcal{W}}-A_{\mathcal{W}})^2]_{ii}\\
&=[D_{\mathcal{W}}^2-D_{\mathcal{W}}A_{\mathcal{W}}-A_{\mathcal{W}}D_{\mathcal{W}}+A_{\mathcal{W}}^2]_{ii}\\
&=\deg^2(v_i)+\deg(v_i).
\end{align*}
Therefore, computing $L^2_{\mathcal{W}}$, we obtain the degrees by solving the quadratic equation above for any $i$.    
\end{remark}

\subsection{Directed Laplacian Identification from Stochastic Inputs}

In this section we address the problem of reconstructing a weighted,
directed Laplacian $L_{\mathcal{D}}$ when the input is a vector of
uncorrelated wide-sense stationary processes $\mathbf{w}\left( t\right) $.

\begin{problem}
\label{Directed Consensus Problem}Consider the dynamical network in %
\eqref{dynamics}, where $\mathcal{D}$ is an unknown weighted, directed
graph. Reconstruct $\mathcal{D}$, from the empirical (cross-)power spectral densities of the
outputs, i.e., $S_{y_{i}}(\omega )$ and $S_{y_{i}y_{j}}(\omega )$ for all $%
1\leq i,j\leq n$.
\end{problem}

In what follows, we propose a reconstruction approach based on a grounded
consensus dynamics, similar to the one proposed in \cite{6203379} for the
reconstruction of undirected graphs. The grounded consensus is defined as
follows:

\begin{definition}[\textit{Grounded Consesus}]\label{groundedcon}
The consensus dynamics with node $v_{j}$ grounded takes the form 
\begin{equation}
\dot{x}(t)=-\tilde{L}_{\mathcal{D}_{j}}x(t)+\mathbf{w}(t)\ ,\ y(t)=x(t),
\label{groundedlaplacian}
\end{equation}%
where $\tilde{L}_{\mathcal{D}_{j}}\in \mathbb{R}^{(n-1)\times (n-1)}$ is
obtained by eliminating the $j$-th row and the $j$-th column from $L_{%
\mathcal{D}}$.
\end{definition}

The consensus dynamics with node $v_{j}$ grounded describes the evolution of
the network when we force the state of node $v_{j}$ to be $x_{j}(t)\equiv 0$%
. We now state a lemma and a corollary to extract $L_{\mathcal{D}}^{T}L_{%
\mathcal{D}}$ from the ungrounded consensus \eqref{dynamics}, and $\tilde{L}%
_{\mathcal{D}_{j}}^{T}\tilde{L}_{\mathcal{D}_{j}}$ from the grounded
consensus \eqref{groundedlaplacian}, respectively.

\begin{lemma}
\label{Lem2}Consider the network dynamics \eqref{dynamics}, when $\mathcal{D}$ is a
weighted, directed network. Let us define the matrix of cross-correlations
as $\mathbf{S}\left( \omega \right) =\left[ S_{y_{i}y_{j}}(\omega )\right] $%
. Then, given assumptions (A1)-(A2), we can compute $L_{\mathcal{D}}^{T}L_{%
\mathcal{D}}$ and $L_{\mathcal{D}}-L^T_{\mathcal{D}}$ as
\begin{align*}
&L_{\mathcal{D}}^{T}L_{\mathcal{D}}=\omega ^{2}\bigg(\frac{\text{Re}\{\mathbf{S}^{-1}(\omega)\}}{[\text{Re}\{\mathbf{S}^{-1}(\omega)\}\mathbf{1}]_{i}} -I \bigg)\\
&L_{\mathcal{D}}-L^T_{\mathcal{D}}=-\omega \bigg(\frac{\text{Im}\{\mathbf{S}^{-1}(\omega)\}}{[\text{Re}\{\mathbf{S}^{-1}(\omega)\}\mathbf{1}]_{i}} \bigg),
\end{align*}%
for any $i$, and for any $\omega $ in the excitation frequency interval $%
\left( -\Omega ,\Omega \right) $.
\end{lemma}

\begin{proof}
Let us consider the matrix $\mathbf{S}(\omega)=[S_{y_{i}y_{j}}(\omega)]$.
According to (\ref{spectrum}), we can separate the real and imaginary parts of
its inverse as%
\begin{align}
&  \text{Re}\{\mathbf{S}^{-1}(\omega)\}=\frac{1}{S_{w}(\omega)}(\omega
^{2}I+L_{\mathcal{D}}^{T}L_{\mathcal{D}})\label{real}\\
&  \text{Im}\{\mathbf{S}^{-1}(\omega)\}=-\frac{1}{S_{w}(\omega)}(\omega
L_{\mathcal{D}}-\omega L_{\mathcal{D}}^{T}).\label{imaginary}%
\end{align}
Post-multiplying \eqref{real} by $\mathbf{1}$, and taking into consideration
that $L_{\mathcal{D}}^{T}L_{\mathcal{D}}\mathbf{1=0}$, we obtain%
\begin{align}
S_{w}(\omega)=\frac{\omega^{2}}{\left[  \text{Re}\{\mathbf{S}^{-1}%
(\omega)\}\mathbf{1}\right]  _{i}}, \label{average2}
\end{align}
for any $i$. Plugging \eqref{average2} into (\ref{real}) and \eqref{imaginary}, the proof follows immediately.
\end{proof}

\begin{corollary}
\label{LDtilde}Consider the network dynamics \eqref{groundedlaplacian}, where $\tilde{L}_{%
\mathcal{D}_{j}}$ is the grounded Laplacian matrix of the weighted, directed
network $\mathcal{D}$ grounded at $v_{j}$. Let us define the matrix of
cross-correlations as $\mathbf{\tilde{S}}\left( \omega \right) =[
\tilde{S}_{y_{i}y_{k}}(\omega )] $ for all $i,k\neq j$. Then, given
assumptions (A1)-(A2), we can compute $\tilde{L}_{\mathcal{D}_{j}}^{T}\tilde{%
L}_{\mathcal{D}_{j}}$ as%
\begin{equation*}
\tilde{L}_{\mathcal{D}_{j}}^{T}\tilde{L}_{\mathcal{D}_{j}}=\omega ^{2}\bigg(\frac{\text{Re}\{\tilde{\mathbf{S}}^{-1}(\omega)\}}{[\text{Re}\{\mathbf{S}^{-1}(\omega)\}\mathbf{1}]_{i}} -I \bigg),
\end{equation*}
for any $i\neq j$, and for any $\omega $ in the excitation frequency
interval $\left( -\Omega ,\Omega \right) $.
\end{corollary}

\begin{proof}
By Corollary \ref{generalized lemma}, it
holds that
\[
\frac{\tilde{S}_{y_{i}y_{k}}(\omega)}{S_{w}(\omega)}=\mathbf{e}_{i}^{T}(\omega
^{2}I-\im \omega (\tilde{L}_{\mathcal{D}_{j}}-\tilde{L}_{\mathcal{D}_{j}%
}^{T})+\tilde{L}_{\mathcal{D}_{j}}^{T}\tilde{L}_{\mathcal{D}_{j}}%
)^{-1}\mathbf{e}_{k}.
\]
Since $\mathbf{\tilde{S}}(\omega)=[\tilde{S}_{y_{i}y_{k}}(\omega)]$, we have that
\begin{align}
\text{Re}\{\tilde{\mathbf{S}}^{-1}(\omega)\}=\frac{1}{S_{w}(\omega)}(\omega^{2}I+\tilde
{L}_{\mathcal{D}_{j}}^{T}\tilde{L}_{\mathcal{D}_{j}}). \label{tilde}
\end{align}
Substituting the expression for $S_{w}(\omega)$ from Lemma \ref{Lem2},
we have 
\[
\tilde{L}_{\mathcal{D}_{j}}^{T}\tilde{L}_{\mathcal{D}_{j}}=\omega^{2}\left(
\frac{\text{Re}\{\tilde{\mathbf{S}}^{-1}(\omega)\}}{\left[  \text{Re}\{\mathbf{S}%
^{-1}(\omega)\}\mathbf{1}\right]  _{i}}-I\right)  .
\]

\end{proof}

In general, the results of Lemma \ref{Lem2} might not be informative enough to extract the underlying structure of the network without running the grounded consensus. We will see at the end of this section how running both ungrounded and grounded consensus will lead to exact reconstruction of the network. However, imposing certain conditions on the network graph, allows us to perform the reconstruction without recourse to grounding. For instance, for a network that does not contain any bidirectional edge, i.e., if $\left(v_{j},v_{i}\right) \in \mathcal{E}_{d}$, then $\left(
v_{i},v_{j}\right) \not\in \mathcal{E}_{d}$, we can employ Lemma \ref{Lem2} to identify the network. The adjacency matrix of a purely unidirectional network satisfies 
\begin{align}
\text{tr}(A^2_{\mathcal{D}})=0. \label{reciprocity}
\end{align}
We now proceed to the identification technique of a purely unidirectional network. 

\begin{corollary}
Suppose the conditions in the Lemma \ref{Lem2} hold. Also, suppose the weighted, directed network satisfies \eqref{reciprocity}. Then, the entries of adjacency are recovered as
\begin{align*}
\left[ A_{\mathcal{D}}\right] _{ij}=\max \bigg\{ \omega n \bigg(\frac{[\text{Im}\{\mathbf{S}^{-1}(\omega)\}]_{ij}}{\mathbf{1}^T \text{Re}\{\mathbf{S}^{-1}(\omega)\}\mathbf{1}} \bigg)   , 0 \bigg\}.\\
\end{align*} 
\end{corollary}

\begin{proof}
For any $k$, by \eqref{average2} we have 
\begin{align*}
[\text{Re}\{\mathbf{S}^{-1}(\omega)\}\mathbf{1}]_k=\frac{1}{n}(\mathbf{1}^T \text{Re}\{\mathbf{S}^{-1}(\omega)\}\mathbf{1}).
\end{align*}
Therefore, under purview of Lemma \ref{Lem2} we obtain
\begin{align*}
L_{\mathcal{D}}-L^T_{\mathcal{D}}=-\omega n \bigg(\frac{\text{Im}\{\mathbf{S}^{-1}(\omega)\}}{\mathbf{1}^T \text{Re}\{\mathbf{S}^{-1}(\omega)\}\mathbf{1}} \bigg).
\end{align*}
If the $ij$-th entry of the right hand side in the above equation is negative, since there is no bidirectional edge, $\left[ A_{\mathcal{D}}\right] _{ij}\neq 0$ and $\left[ A_{\mathcal{D}}\right] _{ji}=0$. If the $ij$-th entry of the right hand side in the above equation is positive, the $ji$-th entry would be negative which implies $\left[ A_{\mathcal{D}}\right] _{ji}\neq 0$ and $\left[ A_{\mathcal{D}}\right] _{ij}=0$. In case the $ij$-th entry is zero, no directed edge between $v_i$ and $v_j$ exists, and the proof is complete. 
\end{proof}

In the next theorem, we show using grounding, there is no need for structural conditions  to solve Problem \ref{Directed Consensus Problem}.

\begin{theorem}
\label{Directed adjacency}Consider the network dynamics \eqref{dynamics},
where $\mathcal{D}$ is a weighted, directed network. Let us also consider
the grounded consensus dynamics in (\ref{groundedlaplacian}), when node $%
v_{j}$ is grounded. Then, given assumptions (A1)-(A2), we can recover the
entries of the weighted, directed adjacency matrix $A_{\mathcal{D}}$ as%
\begin{equation*}
\left[ A_{\mathcal{D}}\right] _{ji}=\left\{ 
\begin{array}{ll}
\sqrt{\left[ L_{\mathcal{D}}^{T}L_{\mathcal{D}}\right] _{ii}-\left[ \tilde{L}%
_{\mathcal{D}_{j}}^{T}\tilde{L}_{\mathcal{D}_{j}}\right] _{ii}}\text{} & 
\text{for }i<j, \\ 
\sqrt{\left[ L_{\mathcal{D}}^{T}L_{\mathcal{D}}\right] _{ii}-\left[ 
\tilde{L}_{\mathcal{D}_{j}}^{T}\tilde{L}_{\mathcal{D}_{j}}\right] _{i-1,i-1}}%
\text{} & \text{for }i> j,%
\end{array}%
\right. 
\end{equation*}%
for any $i\neq j$.
\end{theorem}

\begin{proof}
For simplicity, we consider the case $j=n$ (for any $j\neq n$, we can
transform the problem to the case $j=n$ via a simple reordering of rows and
columns). Running the ungrounded consensus \eqref{dynamics} and applying Lemma
\ref{Lem2}, we recover $L_{\mathcal{D}}^{T}L_{\mathcal{D}}$. Also, running the
grounded consensus \eqref{groundedlaplacian} and applying Corollary
\ref{LDtilde}, we can recover $\tilde{L}_{\mathcal{D}_{n}}^{T}\tilde
{L}_{\mathcal{D}_{n}}$. Hence, for any $i<n$%
\[
\lbrack L_{\mathcal{D}}^{T}L_{\mathcal{D}}]_{ii}-[\tilde{L}_{\mathcal{D}_{n}%
}^{T}\tilde{L}_{\mathcal{D}_{n}}]_{ii}=\sum_{k}\left[  L_{\mathcal{D}}\right]
_{ki}^{2}-\sum_{k\neq n}\left[  L_{\mathcal{D}}\right]  _{ki}^{2}=\left[
L_{\mathcal{D}}\right]  _{ni}^{2}.
\]
Hence, we can recover the adjacency as%
\[
\left[  A_{\mathcal{D}}\right]  _{ni}=\sqrt{\left[  L_{\mathcal{D}}%
^{T}L_{\mathcal{D}}\right]  _{ii}-\left[  \tilde{L}_{\mathcal{D}_{j}}%
^{T}\tilde{L}_{\mathcal{D}_{j}}\right]  _{ii}},
\]
for any $i<n$. The same analysis holds for any other $j\neq n$.
\end{proof}

\subsection{Boolean Reconstruction of General Dynamics}
In this section, we study the reconstruction of 
a general dynamics $G_\mathcal{D}$ in weighted, directed networks, characterized as
\begin{equation}
\dot{x}(t)=G_\mathcal{D} x(t)+\mathbf{w}(t)\ ,\ y(t)=x(t),  \label{general dynamics}
\end{equation}
where $G_\mathcal{D}$ is a negative semi-definite matrix.

It turns out that for the general dynamics \eqref{general dynamics}, complete identification including weights is not possible. However, we can propose an algorithm for a Boolean reconstruction of the system which detects the existence and direction of edges in the directed network \eqref{general dynamics}.

\begin{problem}
\label{General Problem}Consider the dynamical network in %
\eqref{general dynamics}, where $\mathcal{D}$ is an unknown weighted, directed
graph. Reconstruct an unweighted version of $\mathcal{D}$, from the empirical (cross-)power spectral densities of the
outputs, i.e., $S_{y_{i}}(\omega )$ and $S_{y_{i}y_{j}}(\omega )$ for all $%
1\leq i,j\leq n$.
\end{problem}

To solve Problem \ref{General Problem}, similar to Definition \ref{groundedcon}, we define the {\it Grounded Dynamics} at node $v_{j}$ as
\begin{equation}
\dot{x}(t)=\tilde{G}_{\mathcal{D}_{j}}x(t)+\mathbf{w}(t)\ ,\ y(t)=x(t),
\label{groundeddynamics}
\end{equation}%
where $\tilde{G}_{\mathcal{D}_{j}}\in \mathbb{R}^{(n-1)\times (n-1)}$ is
obtained by eliminating the $j$-th row and the $j$-th column from $G_{%
\mathcal{D}}$. We now state a theorem to reconstruct an unweighted version of $\mathcal{D}$. 

\begin{theorem}
\label{Boolean} Consider the network dynamics \eqref{general dynamics},
where $\mathcal{D}$ is a weighted, directed network. Let us also consider
the grounded consensus dynamics in (\ref{groundeddynamics}), when node $%
v_{j}$ is grounded. Then, given assumptions (A1)-(A2), we can recover
a scaled version of off diagonal entires for matrix $G_{\mathcal{D}}$ as
\begin{equation*}
\left[ G_{\mathcal{D}}\right] _{ji}=\left\{ 
\begin{array}{ll}
\sqrt{S_{w}[\text{Re}\{[ \mathbf{S}^{-1}] _{ii}-[\tilde{\mathbf{S}}^{-1}] _{ii}\}}]& \text{for }i<j, \\ 
\sqrt {S_{w}[{\text{Re}\{[ \mathbf{S}^{-1}] _{ii}-[ 
 \tilde{\mathbf{S}}^{-1}] _{i-1,i-1}}\}]}%
 &\text{for }i > j,%
\end{array}%
\right. 
\end{equation*}
for any $i\neq j$.
\end{theorem}
\begin{proof}
Without loss of generality, let $j=n$ as in the proof of Theorem \ref{Directed adjacency}.
By Corollary \ref{generalized lemma}, it holds that
\begin{align*}  
& S_{w}(\omega) \text{Re}\{\mathbf{S}^{-1}(\omega)\}=\omega
^{2}I+G_{\mathcal{D}}^{T}G_{\mathcal{D}}\\
&S_{w}(\omega)\text{Re}\{\tilde{\mathbf{S}}^{-1}(\omega)\}=\omega^{2}I+\tilde
{G}_{\mathcal{D}_{n}}^{T}\tilde{G}_{\mathcal{D}_{n}},
\end{align*}
which entails, for any $i<n$
\begin{align*}
\left[ G_{\mathcal{D}}\right]^2 _{ni}&=[G_{\mathcal{D}}^{T}G_{\mathcal{D}}]_{ii}-[\tilde
{G}_{\mathcal{D}_{n}}^{T}\tilde{G}_{\mathcal{D}_{n}}]_{ii}\\
&=S_{w}(\omega)\bigg([\text{Re}\{\mathbf{S}^{-1}(\omega)\}]_{ii}-[\text{Re}\{\tilde{\mathbf{S}}^{-1}(\omega)\}]_{ii} \bigg),
\end{align*}
which implies the result of the theorem. 
The proof for $j\neq n$ follows with similar arguments.
\end{proof}

\begin{remark}
Results of previous sections regarding Laplacian (Sect. III.A-B), could be extended to any negative semi-definite $G_\mathcal{D}$ with a nonzero nullity, as long as a nonzero vector in the kernel  of $G_\mathcal{D}$ is known.
\end{remark}

Theorems \ref{Undirected Laplacian}, \ref{Directed adjacency} and \ref{Boolean} are
functional for any frequency $\omega \neq 0$, and we do not require the
knowledge of entire power spectral densities. In fact, we only need the
power spectral densities evaluated at one frequency, which dramatically
reduces the computational complexity. One can also evaluate the spectra at
several frequencies to average out the measurement noise and make a robust
reconstruction.

\section{CONCLUSIONS}

In this paper, we addressed the problem of identifying the topology of an unknown,
weighted, directed network running a consensus dynamics when the system is stimulated by a wide-sense stationary noise of unknown power spectral density. We proposed a methodology for network reconstruction based on sequentially grounding nodes in the system. We showed how to reconstruct the topology of the network from the empirical cross-power spectral densities of the outputs between every pair of nodes. We also established that, in the special cases of undirected (i.e., perfect edge reciprocity) or purely unidirectional networks (i.e., no edge reciprocity), our reconstruction procedure does not require any grounding. Finally, we extended our results to the case of a directed network with a general dynamics, and proved that the developed method can detect the presence of edges and their direction.

\section*{ACKNOWLEDGMENTS}
The authors would like to thank Ali Jadbabaie for helpful comments and discussions.

\addtolength{\textheight}{-12cm} 




\bibliographystyle{unsrt}
\bibliography{shahin}

\end{document}